\documentclass[conference]{IEEEtran}
\IEEEoverridecommandlockouts
\usepackage{cite}
\usepackage{amsmath,amssymb,amsfonts}
\usepackage{textcomp}
\usepackage{xcolor}
\usepackage[linesnumbered,ruled,vlined]{algorithm2e}
\usepackage{soul} 
\usepackage{graphicx, tcolorbox}
\usepackage{makecell}
\usepackage{subcaption}
\usepackage{booktabs}
\usepackage{url}
\usepackage{hyperref}
\usepackage{cleveref}      
\usepackage[font=small]{caption}

\usepackage{amsthm}
\newtheorem{lemma}{Lemma}
\newtheorem{theorem}{Theorem}

\def\BibTeX{{\rm B\kern-.05em{\sc i\kern-.025em b}\kern-.08em
    T\kern-.1667em\lower.7ex\hbox{E}\kern-.125emX}}
\begin{document}

\title{Fully Dynamic Rooted Spanning Tree on GPU}

\author{
\IEEEauthorblockN{Abhijeet Sahu, 
Harmit Singh,
Soham Nandy,
G. Ramakrishna}

\IEEEauthorblockA{Indian Institute of Technology Tirupati, India, 517619\\
\{cs22s501, cs20b012, cs20b046, rama\}@iittp.ac.in}}

\maketitle

\begin{abstract}
Spanning trees are fundamental structures in graph theory, essential for various applications such as network maintenance, routing adjustments, and many more. 
The dynamic nature of real-world networks requires efficient updates to these structures as the underlying graph evolves. Maintaining rooted spanning trees dynamically is particularly crucial for algorithms addressing 2-connected components and minimum-weighted spanning trees. In this paper, we address the challenge of maintaining a rooted spanning forest when a batch of edges are inserted or deleted. We present four novel fully dynamic parallel algorithms to update the spanning forest without reconstructing it from scratch. To the best of our knowledge, parallel algorithms for this problem remain largely unexplored. Our experiments on a diverse collection of real-world graphs using a \textsc{gpu} environment demonstrate a throughput of 2 million insertions and 1.4 million deletions per second, significantly outperforming state-of-the-art parallel static algorithms.
\end{abstract}

\begin{IEEEkeywords}
Dynamic graph, fully dynamic, batch updates, spanning tree
\end{IEEEkeywords}

\section{Introduction}
The efficient processing of large graphs is increasingly vital across various computational domains, as recent studies underscore the growing complexity and scale of networks~\cite{sakr2021future}. The rise of online social networks and other dynamic environments demands rapid computations of various graph metrics~\cite{sajal_das_mst}. Real-world networks, such as those in transportation, biology, and social platforms, experience continuous changes or $``churn"$. For instance, on a popular social media platform $``Reddit"$, users frequently join and leave $subreddits$, leading to constant additions and deletions of connections within the underlying network.
These frequent updates pose a significant challenge: efficiently updating vital graph metrics as the graph undergoes changes. 

Traditional methods of analysing static snapshots of these networks fall short, given the rapid pace of these changes. This problem has created a new demand for the development of algorithms tailored for dynamic graphs, where edges can be frequently inserted or deleted. A common approach is to accumulate updates into batches that can be processed in parallel. This approach improves performance as it reduces redundant computations and increases parallelism, although it introduces its own set of challenges.

\noindent Spanning trees are fundamental to many graph-related applications, playing a critical role in fields like telecommunications, the Internet of Things, transportation, and more. Many crucial network tasks, such as database maintenance, can be executed efficiently using a tree that spans the network. Maintaining spanning trees efficiently in dynamic networks is more practical and resource-efficient than rebuilding them from scratch after each update. This approach not only accelerates tree construction, but also significantly improves the performance of dependent dynamic algorithms such as biconnected components (\textsc{bcc}), minimum weight spanning Tree (\textsc{mst}), centrality measures, and others.

For example, \textbf{updating a rooted spanning tree} is required as a subroutine in dynamic batch updates of \textsc{mst} \cite{sajal_das_mst} and biconnected components \textsc{bcc} \cite{chirayu_bcc} algorithms. However, the current \textsc{mst} algorithm resorts to completely rebuilding the tree, even when updates to an existing rooted spanning tree would suffice. In the dynamic \textsc{bcc} algorithms, the dynamic maintenance of a rooted spanning tree is required. While some level of parallelism is achieved by processing multiple broken trees concurrently, but the path reversal within each tree is done serially by a single thread. Although this coarse-grained parallelism suits multicore \textsc{cpu}, it is not ideal for modern \textsc{gpu}s, where fine-grain parallelism is expected. Additionally, the unique memory hierarchy of \textsc{gpu}s, coupled with the irregular data access patterns in graphs, necessitates careful design of data layouts and memory accesses. Thus, designing graph algorithms for \textsc{gpu}s requires meticulous attention to both parallelism and memory management.


In this paper, we propose \textsc{gpu} based parallel algorithms to  dynamically update  a rooted spanning tree or forest following the insertion or deletion of a batch of edges. We now formally define the problem.


\begin{tcolorbox}[colframe=black, colback=white, sharp corners, boxrule=0.4mm, width=\linewidth,
                  title=\textbf{Dynamic Rooted Spanning Forest }(\textsc{d-rsf})]
\textbf{Input}: An undirected graph $G=(V,E)$, a rooted spanning forest $F$ of $G$, a batch $B$ of edges, 
and an operation $\mathsf{op} \in \{ \mathsf{insert}, \mathsf{delete} \}.$

\textbf{Output}: A rooted spanning forest $F'$ of the updated graph $G'$, where
\[
  G' = 
    \begin{cases}
      G + B & \text{if }\text{op}=\text{insert},\\
      G - B & \text{if }\text{op}=\text{delete}.
    \end{cases}
\]
\end{tcolorbox}

Given a graph $G$ with $V$ vertices and $E$ edges, and a rooted spanning forest $F$ of $G$, and a batch $B$ of edges to be inserted or deleted according to an operation $\mathbf{op} \in \{\mathsf{insert}, \mathsf{delete}\}$ the goal is to efficiently maintain \( F \) through these dynamic updates. Specifically, the aim is to update the spanning forest dynamically from the current solution, rather than reconstructing the spanning forest from scratch after each operation. 

\subsection{Related Work}
Over the decades, spanning trees have been widely studied across various computational models, with recent focus shifting towards dynamic spanning trees due to the need for efficient algorithms capable of handling dynamic updates. \\
\noindent \textbf{Static Algorithms.} Merrill et al. \cite{merrill2012scalable}, proposed a work-efficient BFS algorithm using the Prefix Sum technique. While BFS-based methods are efficient for graphs with small diameters, they suffer from performance degradation as the graph depth increases \cite{emc}.

Beyond BFS, different parallel algorithms for constructing spanning trees have also been explored. Cong et al. \cite{pr_rst} proposed a multicore algorithm based on the parallel Union-Find technique, combining grafting and broadcasting to reverse paths and construct a rooted spanning tree. However, the GPU implementation of \textsc{pr-rst} algorithm experiences computational overhead due to multiple rounds of pointer jumping, leading to increased complexity and irregular memory access patterns.

A fundamental limitation shared by all static approaches is that they do not exploit the existing spanning forest when the graph is updated. Even for small changes (e.g., 1--2\% of edges), a full recomputation over the entire graph is required, making them ill-suited for dynamic settings.

\noindent \textbf{Sequential Dynamic Algorithms.}
 To address this issue, Eppstein et al. proposed the \emph{Sparsification Technique}, a widely recognized paradigm in sequential settings for designing dynamic graph algorithms \cite{eppstein1997sparsification}, applicable to connected components, biconnected components, and minimum weight spanning trees. Nevertheless, the sparsification technique has notable limitations, including its inherent sequential nature and its focus on single-edge updates. Although this approach has been successfully adapted for dynamic connected components \cite{srinivasan2016application}, it proved unsuitable for maintaining minimum weight spanning trees in multicore architectures \cite{sajal_das_mst}, likely due to associated overheads. Furthermore, empirical studies by Haryan et al. \cite{chirayu_bcc} demonstrated that the sparsification technique falls short in practice compared to parallel static algorithms for maintaining biconnected components. The exploration of the applicability of the sparsification paradigm to GPU architectures still remains limited, presenting an open research avenue.

\noindent \textbf{Parallel Dynamic Algorithms.}  
To overcome the limitations of sequential dynamic algorithms, parallel batch dynamic algorithms have been explored for maintaining spanning forests under batch updates. An algorithm is referred to as \emph{fully dynamic} if it handles both edge insertions and deletions, whereas \emph{incremental} and \emph{decremental} algorithms support only insertions or deletions, respectively.  
To the best of our knowledge, three notable contributions on fully dynamic spanning forests are \cite{hipc_2013, spaa_2019, archiv_2024}.

The HDT (Holm, de Lichtenberg, and Thorup) algorithm \cite{spaa_2019} and the cluster forest algorithm \cite{archiv_2024} provide theoretical guarantees on algorithmic depth, but their practical adaptation to specific architectures such as multi-core processors or \textsc{gpu}s remains unexplored. In particular, adapting these algorithms to \textsc{gpu} architectures is non-trivial and remains an open research challenge. A key limitation of the HDT algorithm is its requirement to maintain $O(\log V)$ spanning forests within a hierarchical data structure, leading to a space complexity of $O(V \log V + E)$.

The cluster forest algorithm \cite{archiv_2024} performs two independent graph traversals in parallel from the endpoints of each deleted edge. While this kind of parallelism is effective on multi-core architectures, it is not well-aligned with the parallel execution model of \textsc{gpu}s due to their architectural constraints. Similarly, the parallel algorithm proposed in \cite{hipc_2013}, though implemented for multi-core systems, lacks theoretical guarantees on depth. This algorithm is based on a static parallel \textsc{bfs}, with its worst-case depth bounded by $O(V + E)$, where $V$ and $E$ denote the number of vertices and edges in the graph.

In summary, existing approaches either prioritize theoretical guarantees or focus on architecture-specific optimizations, but not both. In contrast, our work aims to bridge this gap by designing algorithms that offer theoretical guarantees while being well-suited for practical deployment on \textsc{gpu} architectures.
A comprehensive comparison of existing works with our proposed approach is presented in \textbf{Table ~\ref{table:comparison}}.

\begin{table}[h!]
\centering
\scriptsize
\caption{Depth comparison of algorithms for batch edge operations in dynamic graph algorithms, with \textsc{sg-et} and \textsc{hs-et} representing our proposed contributions. Here, $V$, $E$, and $k$, denote the number of vertices, edges, and trees in the underlying forest respectively.}
\begin{tabular}{@{}lccc@{}}
\toprule
\textbf{Algorithms} & \textbf{Batch-Insert} & \textbf{Batch-Delete} & \textbf{Space} \\ \midrule
\textbf{HDT}~\cite{spaa_2019}          & $O(\log V)$              & $O(\log^3 V)$            & $O(V\log V + E)$ \\
\textbf{CF}~\cite{archiv_2024} & $O(\log^3 V)$            & $O(\log^3 V)$            & $O(V + E)$ \\
\textbf{SG-ET} (ours) & $\boldsymbol{O(\log V + \log^2 k)}$ & $\boldsymbol{O(\log V + \log^2 k)}$ & $\boldsymbol{O(V + E)}$ \\
\textbf{HS-ET} (ours) & $\boldsymbol{O(\log^2 V)}$          & $\boldsymbol{O(\log^2 V)}$          & $\boldsymbol{O(V + E)}$ \\ \bottomrule
\end{tabular}
\label{table:comparison}
\end{table}

\noindent \textbf{Past work on Dynamic Trees.}
Dynamic spanning tree problems and dynamic trees are related, but they have different goals. Dynamic trees specifically deal with maintaining a forest of trees, and support insert and delete operations on edges, while ensuring the structure remains \textbf{acyclic}. In contrast, dynamic spanning tree problems operate on general graphs, maintaining a spanning tree while accommodating edge insertions and deletions, making them inherently more complex.

The concept of dynamic trees was pioneered by Sleator and Tarjan ~\cite{tarzan_link_cut} (1981) with the introduction of link/cut trees for fast network flow applications, later adapted by Frederickson ~\cite{frederickson1983data} (1983) for dynamic MST updates. Following that, subsequent researchers have expanded this area extensively in both sequential and parallel settings~\cite{henzinger1999randomized,tseng2019batch}.

\subsection{Our Contributions}
Designing parallel dynamic algorithms for \textsc{d-rsf}, where the algorithm’s depth depends on the number of affected key edges ($m$) or vertices ($n$), poses significant challenges. The algorithms presented in this paper directly address these challenges. Our design focuses on achieving low-depth parallel algorithms to ensure practical efficiency, particularly on \textsc{gpu} architectures. To the best of our knowledge, there are no existing \textsc{gpu}-based parallel algorithms for maintaining rooted spanning forests under dynamic updates. Our key contributions are summarized below:

\begin{itemize} \itemsep0em

\item We design and implement four fully dynamic parallel algorithms for the \textsc{d-rsf} problem, specifically tailored for \textsc{gpu} architectures. We formally prove the correctness of our algorithms and demonstrate their practical efficiency in practical large-scale graph settings.

\item We design a \textsc{gpu}-based parallel subroutine to reverse multiple paths in a rooted spanning forest, and integrate ideas from Hooking-Shortcutting, Broadcasting and Euler-Tour techniques to solve the \textsc{d-rsf} problem.

\item Our \textsc{gpu}-based parallel algorithms significantly outperform existing state-of-the-art static parallel \textsc{gpu} algorithms, achieving speedups of up to 500$\times$ for deletions and up to 900$\times$ for insertions.

\item Beyond maintaining rooted spanning forests, our proposed algorithms can also be applied to maintain the connected components of a graph under dynamic updates, including both edge insertions and deletions.

\end{itemize}
\subsection{Preliminaries.}
In this section we begin by discussing key existing algorithms that provide the necessary context for understanding our algorithm.

\noindent\textbf{Hooking and Shortcutting (\textsc{hs})}\cite{hooking_shortcutting}. This is a parallel algorithm inspired by the union-find algorithm for computing connected components and a spanning forest in logarithmic time.
\textcolor{black}{This algorithm receives a sequence of edges as input. Conceptually, \textsc{hs} maintains a forest of rooted trees that represent the components identified so far. Each parallel round of the algorithm comprises of two alternating phases, namely Hooking followed by Shortcutting.  In the Shortcutting phase, every vertex in a tree ``jumps'' two levels up its tree (pointer doubling), effectively collapsing paths toward the root and flattening the tree structure. In the hooking step, all the intra-edges of rooted trees are discarded. Later, for each edge $(u,v)$ that connects two distinct trees, one tree’s root is made a child of the other, thereby merging components. By careful examination, the formation of cycles is avoided. The edges that were not successfully considered, due to parallelism and race condition will participate in the subsequent rounds. The algorithm alternates between two operations till all the edges are processed successfully. }

\noindent \textbf{Euler Tour.}\cite{emc}. 
\textcolor{black}{For an unrooted tree, a directed graph $T'$ can be obtained by replacing each edge $(u,v)$ with two directed edges: $(u,v)$ and $(v,u)$. We define the Euler tour of such an unrooted tree as a path of directed edges in $T'$ that visits each edge exactly once and returns to the starting vertex. 
For an Euler tour on a sequence $(e_1, e_2 \ldots, e_{2\times (n-1)})$ of edges, for each $1 \leq i \leq 2\times (n-1)$, $rank$ of $e_i$ is defined as its index $i$.
We use a well-known parallel algorithm technique called the \emph{Euler Tour Technique}~\cite{emc} to transform an unrooted tree into a rooted tree, constructing an Euler tour, and computing the ranks of all the edges in of Euler tour.}

\section{Generic Algorithm for D-RSF}
In this section, we begin by introducing the necessary notations that are used throughout the description of our \textsc{d-rsf} algorithm. We then present an overview of the algorithm, followed by two important subroutines to address the subproblems of \emph{oriented replacement edges} and \emph{path reversal}. Finally, we describe the complete generic algorithm using the introduced notation and subroutines.

\subsection{Notation}{
Let $G=(V(G),E(G))$ be an input undirected graph, where $V(G)$ and $E(G)$ denote the set of vertices and edges, respectively in $G$.
A forest $F = (V(F), E(F))$ of $G$ is said to be a \emph{rooted spanning forest} of $G$ if $V(F) = V(G)$ and for every vertex $v$ in $F$, there exists a parent except for one vertex in each component. The vertices that do not have parents correspond to \emph{root vertices}. For an edge $e = (x, y)$ in a rooted tree, where $y$ is the parent of $x$, we say that the \emph{orientation} of $e$ is from $x$ to $y$ $(x \rightarrow y)$.

\begin{table}[ht]
\centering
\small
\caption{Notations}
\begin{tabular}{|c|l|}
\hline
\textbf{Notation} & \textbf{Description} \\ \hline
$G$ & An undirected graph \\ \hline
$V(G)$, $n$ & The number of vertices in $G$ \\ \hline
$E$ & The number of edges in $G$\\ \hline
$E(G)$ & The set of edges in $G$ \\ \hline

$F$ & A rooted spanning forest of $G$ \\ \hline
$x \rightarrow y$ & \makecell[l]{Orientation of an edge from child $x$ \\ to its parent $y$ in $F$} \\ \hline
$m$ & The number of key edges \\ \hline
$k$ & The number of components (trees) in $F$\\ \hline

$G'$ & Graph $G$ after a batch update \\ \hline
 $F'$&A rooted spanning forest of $G'$\\ \hline

$\hat{F}$ & \makecell[l]{
$\hat{F} = F - B$ for deletion operation;\\
$\hat{F} = F$ for insertion operation
} \\ \hline

\end{tabular}
\label{table:notations}
\end{table}

Let $F$ be a rooted spanning forest of $G$ on $n$ vertices.
The resultant graph obtained after inserting or deleting a batch $B$ of edges in $G$ is denoted by $G'$.
We use $M$ to denote a set of potential edges that help to establish the connectivity across the trees in a spanning forest of $G'$ and such edges are referred as \textbf{key edges}.
In case of delete operation, $M= E(G) - E(F) - B$, where as $M = B$ for insert operation. We use $m$ to denote the number of key edges.

For delete operation, $\hat{F}$ denote the forest obtained after deleting edges of $B$ from $F$, whereas $\hat{F}$ remains same as $F$ for insert operation.
Each tree $T_i$ in $\hat{F}$ is uniquely identified by its root vertex $r_i$. For each vertex $v$ in $\hat{F}$, the root vertex of the tree containing $v$ is called as \emph{representative}, which is denoted by $rep[v]$.
An edge $(u, v) \in M$ is a \emph{cross edge} if $rep[u] \neq rep[v]$. Let $E_{cross}$ be the set of all cross edges in $G'$ with respect to $\hat{F}$. 
An edge set $E_r \subseteq E_{cross}$ is referred to as \emph{replacement edges} if $|E_r|=k-\ell$, and at most one edge from $E_r$ appears between any two trees in $|\hat{F}|$, where $k$ and $\ell$ denotes the number of components in $\hat{F}$ and $G'$, respectively.}

\begin{figure*}[htbp]
\centering
\includegraphics[width=0.9\linewidth]{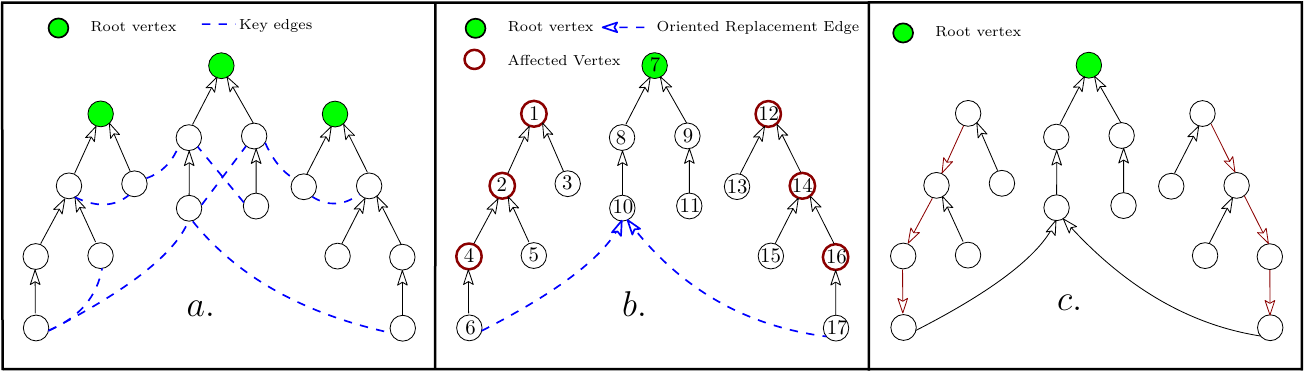}
\caption{Overview of the \textsc{d-rsf} algorithm. 
(a) The black solid lines represent the rooted spanning forest, while the blue edges denote the key edges. 
(b) Replacement edges are selected from the key edges and oriented appropriately. 
(c) Parents are updated for affected vertices during path reversal to obtain an updated rooted spanning forest.}
\label{overview_dig}
\end{figure*}

\subsection{Overview}
We now present a high-level overview of our generic parallel algorithm to solve the \textsc{d-rsf} problem. The dynamic algorithm updates the rooted spanning forest of a graph in parallel, based on a batch of edges to be inserted or deleted, without reconstructing it from scratch.

First, depending on the type of update operation, we either insert the given batch of edges into the underlying graph or remove them from both the forest and the underlying graph. 
Next, we identify \emph{replacement edges} from the available set of \emph{key edges} to reestablish connectivity across the trees in the forest. Once the replacement edges are identified, assigning their orientations arbitrarily can lead to conflicts where a vertex may end up with multiple parents. For example, in Figure~\ref{overview_dig}(b), both vertices $6$ and $17$ could attempt to become the parent of vertex $10$ if orientations are assigned arbitrarily. Such conflicts can result in the loss of valid orientations for certain edges. Therefore, it is necessary to systematically determine the \emph{orientations} of the replacement edges before including them in the forest. The task of finding \emph{oriented replacement edges} constitutes our first subproblem.

Although it is possible to assign orientations to replacement edges without causing internal conflicts, these new orientations may still overwrite the parent relationships of existing tree edges. For instance, in Figure~\ref{overview_dig}(b), the addition of a replacement edge $(6, 10)$ could reassign the parent of vertex $6$ to vertex $10$, thereby causing the loss of the original tree edge $(4, 6)$.

To overcome this, we introduce a second subproblem, which involves reversing the directions of a few tree edges to maintain a valid rooted forest. More precisely, we reverse the orientations along certain paths (e.g., the path from vertex $6$ to $1$ in Figure~\ref{overview_dig} (b)), and then later insert the replacement edges to obtain the new forest.
Overall, the \textsc{d-rsf} problem is divided into two sub-problems. The first sub-problem involves finding the \textbf{orientated replacement edges} across the disconnected trees, while the second aims to reverse the orientations of multiple paths in a forest, shortly called \textbf{reverse paths}.

\subsection{Subroutines}
To address the two sub-problems proposed earlier, we now present the necessary subroutines. We defer the algorithm details required for these subroutines to later subsections. 

\noindent
\textcolor{black}{\textbf{\textsc{FindOrientedReplEdges}}($rep[\cdot]$, $M$)}: \label{FindReplEdges}
\textcolor{black}{Given the representative array $rep[\cdot]$ for a rooted forest $\hat{F}$ and a set $M$ of key edges, this subroutine identifies oriented replacement edges from $M$ to establish connectivity across the trees in $\hat{F}$.}

\noindent 
\textcolor{black}{\textbf{\textsc{ReversePaths}}($R$, $parent[\cdot]$))}:
\textcolor{black}{Given a rooted forest $\hat{F}$ (represented using the $parent[\cdot]$) and a set $R$ of pairs, where each pair denotes the end vertices of a path in $\hat{F}$. This subroutine reverses the parent-child relationships of all vertices along the paths in $R$ in parallel. }

\begin{algorithm}[htbp]
    \DontPrintSemicolon 
    \caption{A GPU Algorithm for Dynamic Rooted Spanning Forest} \label{alg:main_algo}
    \SetKwInput{KwInput}{Input}
    \SetKwInput{KwOutput}{Output}
    \KwInput{A Graph $G = (V,E)$, a rooted spanning forest $F$ of $G$ represented by $parent[\cdot]$, a batch of edges $B$, and an operation $\mathsf{op} \in \{ \mathsf{insert}, \mathsf{delete} \}$.}
    
    \KwOutput{A Rooted Spanning Forest $F'$ of $G' = G\ \mathsf{op}\ B$}
    
    \SetKwBlock{DoParallel}{do in parallel}{end}
 
    $M = B$\; \label{line:insert_m}
 
    \If{operation is delete \label{line:update_ds}} {
        $M = (E(G) - E(F)) - B$\; \label{line:compute_m}
        \For{each edge $(u, v) \in B$ \textbf{in parallel}}{
            \tcc{Let $v$ be the $parent$ of $u$}
            $parent[u] \gets u$ \; \label{line:parent_update}
        }
    }

    $rep[\cdot] \gets$ \textbf{parallel pointer jumping} on $parent[\cdot]$ \; \label{line:pointerJumping}

    \textcolor{black}{$E' \gets$ \textsc{FindOrientedReplEdges}($M$, $rep[\cdot]$)} \label{line:ReplacementEdges}\;
    \For{each oriented edge $(u, v) \in E'$ \textbf{in parallel}}{
        \tcc{Let $v$ be the $parent$ of $u$} 
        $R = R \cup \{(u, rep[u])\}$\; \label{line:GenerateInterval}
    }
    $F' \gets$ \textsc{ReversePaths}($\hat{F}$, $R$, $parent$)\; \label{line:PathReversal}
    Update $parent[\cdot]$ in parallel for each edge in $E'$\; \label{line:IncludeNewParents}
    Insert/Delete $B$ of edges to/from $G$\;
\end{algorithm}

\subsection{Generic Algorithm}

An instance of the \textsc{d-rsf} problem involves an undirected graph $G$, a rooted spanning forest $F$ of $G$, and a batch $B$ of edges for insertion or deletion.  Algorithm~\ref{alg:main_algo} begins by identifying the operation type. For deletions, edges in $B$ are removed from the non-tree edges, and the key edges $M$ are computed in parallel (Line~\ref{line:compute_m}).
Subsequently, for each tree edge in $B$, the forest $F$ is updated in parallel by setting the parent node to itself (Line \ref{line:parent_update}), effectively deleting all the tree edges in $B$ from the forest.
For insert operation, since no non-tree edge exists between disconnected trees in the forest, $B$ itself becomes key edges in Line~\ref{line:insert_m}.
After updating the forest and obtaining key edges, the algorithm proceeds to determine the representative of each vertex, resulting in the $rep[\cdot]$, by applying pointer jumping on the parent array (Line~\ref{line:pointerJumping}). 
This ensures that each vertex in the forest can retrieve its root vertex quickly in the subsequent tasks. 
The algorithm then identifies oriented replacement edges $E'$ from the set $M$ of key edges using the parallel subroutine \textsc{FindOrientedReplEdges} (Line~\ref{line:ReplacementEdges}). 
\textcolor{black}{For instance, As illustrated in Figure~\ref{overview_dig}, it identifies two replacement edges: $(6,10)$ and $(10,17)$. However, when including an edge $(6,10)$ and setting its orientation such that $10$ is a parent of $6$, it is necessary to reverse the orientations of the edges involved in the path from $6$ to $1$. }Therefore, before incorporating the oriented replacement edges into the forest, we first identify the end vertices of the paths whose orientations must be reversed (Line~\ref{line:GenerateInterval}). The paths represented by pairs in $R$ are then reversed using the parallel subroutine \textsc{ReversePaths} (Line~\ref{line:PathReversal}). Finally, the new replacement edges are incorporated into the updated forest $F'$ (Line~\ref{line:IncludeNewParents}) using $|E'|$ threads, ensuring that the spanning forest remains properly rooted after edge updates.

We now present the proof of correctness for the generic algorithm described in Algorithm~\ref{alg:main_algo}).

\begin{theorem}
Given an undirected graph $G$, a rooted spanning forest $F$ of $G$, a batch of edges $B$ to insert or delete, Algorithm~\ref{alg:main_algo} produces a valid rooted spanning forest $F'$ of $G \pm B$.
\end{theorem}

\begin{proof}
Let $G'$ be the graph and $\hat{F}$ the forest obtained after the batch updates. Let $\ell$ denote the number of connected components in $G'$ and $k = |\hat{F}|$ represent the number of trees in $\hat{F}$. Using the function \textsc{OrientedReplEdges} (Line~\ref{line:ReplacementEdges}), we compute the set of oriented replacement edges $E'_r = \{e'_1, e'_2, \ldots, e'_{k-\ell}\}$. 

Let $e_k = (x, y)$ be a $replacement~edge$ in $E'_r$. Consider two trees $T_i$ and $T_j \in \hat{F}$ rooted at $r_i$ and $r_j$, respectively, where $x \in T_i$ and $y \in T_j$. The simple path from $x$ to the root $r_i$ is denoted as $P = \{x, v_1, v_2, \ldots, r_i\}$. After applying $\textsc{ReversePaths}$ (Line~\ref{line:PathReversal}), the parent-child relationships along $P$ are reversed, resulting in every vertex $v \in T_i$ having a parent, except for $x$. When the edge $(x, y)$ is inserted, $x$ acquires $y$ as its parent in Line~\ref{line:IncludeNewParents}.

 The addition of each $replacement~edge$ reduces the number of root vertices in $\hat{F}$ by one. Starting with $k$ roots, the insertion of $k-\ell$ replacement edges decreases the root count to $k - (k-\ell) = \ell$. Therefore, the resulting forest $\hat{F}$ is a valid rooted spanning forest of $G'$ with exactly $\ell$ trees, one per connected component of $G'$.
\end{proof}

\section{Replacement Edges and Path Reversal}
The two main subproblems in our generic algorithm to solve \textsc{d-rsf} are finding oriented replacement edges and reversing multiple independent paths, both in parallel. We propose Supergraph (\textsc{sg}) and Hooking Shortcutting (\textsc{hs}) techniques  to address the first problem. Further, we describe Broadcasting (\textsc{bc}) \cite{pr_rst} and design Euler Tour (\textsc{et}) based technique to solve the second. 
\textcolor{black}{Since there are two distinct approaches to solving two subproblems mentioned, we propose the following four algorithms to solve \textsc{d-rsf}.}

\begin{enumerate}
 \item \textsc{sg-et} (Supergraph with Euler Tour)
 \item \textsc{sg-bc} (Supergraph with Broadcasting)
 \item \textsc{hs-et} (Hooking Shortcutting with Euler Tour)
 \item \textsc{hs-bc} (Hooking Shortcutting with Broadcasting)
\end{enumerate}

\subsection{Supergraph for Oriented Replacement Edges}
\label{section_SGorientedReplacementEdges}
Given a forest $\hat{F}$, a representative array $rep[\cdot]$ to hold representative for every vertex in  $\hat{F}$ and a set $M$ of key edges, we now identify oriented replacement edges. 
In this supergraph (\textsc{sg}) approach we construct a supergraph $\tilde{G}$ and a hash table using the forest and key edges. The supergraph helps to maintain the topological structure across the trees in the forest, whereas the hash table is useful to retrieve a key edge corresponding to any edge in the supergraph. 
We first construct a supergraph $\tilde{G}$, in which each vertex corresponds to a tree in the forest $\hat{F}$ and an edge in the supergraph corresponds to a cross edge across the trees in the forest. 
In particular, for each root vertex $u$ in $\hat{F}$ there is a  vertex in $\tilde{G}$. We go through all key edges $e=(u,v)$ in parallel and the edge $(u,v)$ is marked as \emph{cross edge} if $rep[u]\neq rep[v]$. We now use a temporary array $sedges$ to capture the potential edges to be added in the supergraph.
For each cross edge $e_i=(u,v)$ in parallel, we include an edge in $\tilde{G}$ by assigning $sedges[i] = (rep[u], rep[v])$, and insert a key-value pair in a hash table $H$, where $key=(rep[u],rep[v])$, and $value=(u,v)$. At this stage, the number of edges in the supergraph equals to the number of cross edges.  Multiple cross edges over the same two trees result in duplicate edges in the supergraph. To eliminate such redundant edges, we apply parallel sorting and parallel compaction on $sedges$.
Although multiple  pairs with same key attempt to write in the hash table in parallel, eventually no two pairs in the hash table have same key, due to race conditions. In other words, for each  edge in the supergraph $\tilde{G}$, we maintain a unique cross edge in the hash table.

We now apply \textsc{pr-rst} algorithm \cite{pr_rst} on the super graph $\tilde{G}$ to obtain a rooted spanning forest $\tilde{F}$ of $\tilde{G}$. 
By treating the edges of $\tilde{F}$ as keys, we can retrieve the corresponding cross edges from the hash table and these cross edges become oriented replacement edges.
For each edge oriented $(u',v')$ in $\tilde{F}$, where $u'$ is the parent of $v'$, we retrieve a oriented cross edge $(u,v)$ from the hash table $H$, where $u$ is par[$v$]. 

\begin{lemma}
\label{lemma: SG_lemma}
    Given a forest $\hat{F}$ with $k$ trees and $m$ key edges, oriented replacement edges can be obtained using the supergraph approach in $O(\log m + \log^2 k)$ depth and $O(m + k \log k)$ work in the worst case.
\end{lemma}

The merit of this approach is that the depth of this algorithm is independent on the number of vertices $n$, and only depends on the number of disconnected components and the key edges. However, $m$ can be large in the case of deletion operation, and thus we aim to design a parallel algorithm that minimises the dependency on $m$ in the next subsection.

\subsection{Hooking and Shortcutting for Oriented Replacement Edges}

\begin{algorithm}[htbp]
    \DontPrintSemicolon 
    \caption{Oriented Replacement Edges using Hooking and Shortcutting} \label{alg:hs_algo}
    \SetKwInput{KwInput}{Input}
    \SetKwInput{KwOutput}{Output}
    \KwInput{key edges $M[\cdot]$, and representative array $rep[\cdot]$}
    \KwOutput{At most $k - 1$ oriented replacement edges}
    
    \SetKwBlock{DoParallel}{do in parallel}{end}
    \SetKwRepeat{Do}{do}{while}

    $replEdges \gets $ \textsc{Hooking-Shortcutting}$(M, rep)$\; \label{line:rep_edges_comp}

    \For{each $e_i = (u, v) \in replEdges$ \textbf{in parallel}}{ \label{line:hashEdges}
        $sfEdges[i] = (rep[u], rep[v])$\;
        $H$.insert($\langle rep[u], rep[v] \rangle$, $\langle u, v \rangle$)\;
    }
    $\mathcal{O} \gets \textsc{Eulerian-Tour}(sfEdges)$\; \label{line:euler_forest}
    $R \gets$ $H$.retrieve($\mathcal{O}$)\; \label{line:retreive_sf_edges}
\end{algorithm}

Given the set of key edges $M[\cdot]$ and the representative array $rep[\cdot]$ for a rooted forest $\hat{F}$, our method proceeds in two phases. 
In \textbf{Phase 1 (Line~\ref{line:rep_edges_comp})}, we apply the \textsc{Hooking-Shortcutting} (\textsc{hs}) procedure to identify at most $k-1$ replacement edges that can reconnect the disconnected components. 

\textcolor{black}{In the \textsc{hs} algorithm presented in \cite{hooking_shortcutting}, each vertex begins as a singleton tree, whereas in our adaptation we begin with an existing forest of partial components. We then apply the same \textsc{hs} steps to efficiently identify replacement edges.}

In \textbf{Phase 2 (Line~\ref{line:hashEdges} - Line~\ref{line:retreive_sf_edges})}, we orient the selected replacement edges.
To maintain the topological structure across the trees, we first construct an auxiliary forest by storing, for each replacement edge $(u,v)$, a key-value mapping in a hash table $H$, where the key is $\langle rep[u], rep[v] \rangle$ and the value is $\langle u, v \rangle$ (Line~\ref{line:hashEdges}). This auxiliary forest, formed by the keys in $H$, represents a super-forest over the component representatives.
We then apply the Eulerian Tour algorithm to this super-forest (Line~\ref{line:euler_forest}) to obtain a rooted structure. Finally, we retrieve the corresponding original edges from the hash table $H$ (Line~\ref{line:retreive_sf_edges}) to generate the final set of \emph{oriented replacement edges}. For hash table implementation, we adapted an existing open-source code available online \footnote{\url{https://nosferalatu.com/SimpleGPUHashTable.html}}, and modified it to suit our use case.

\begin{lemma}
\label{lemma: HS_lemma}
Algorithm 2 identifies oriented replacement edges in a graph using Hooking and Shortcutting operations in $O(\log^2 n)$ depth and $O((m + n) \log n)$ work.
\end{lemma}

A key merit of this approach is that its depth depends solely on the number of vertices, not the number of edges. While the \textsc{hs} technique is primarily used to identify oriented replacement edges, it serves as a parallel static algorithm for obtaining a rooted spanning tree.

\subsection{Broadcasting Based Path Reversal} \label{BC_path}

The purpose of this subroutine is to reverse the orientation of all disjoint paths of a forest, which are specified by $R$, where each path is represented by a pair $(u_i, r_i)$. We use the broadcasting method by Cong et al. \cite{pr_rst} to efficiently reverse multiple paths in a forest in parallel. The algorithm begins by identifying all the \emph{on-path} vertices that lie on the path from $u_i$ to $r_i$ for each pair in $R$ simultaneously. Once these vertices are identified, the algorithm reverses their parent-child relationships in parallel i.e., if $u = parent[v]$, we now set $parent[u] = v$.

For each vertex $u$ in the tree, an array of size $O(\log n)$, referred to as the \emph{special ancestor}, is maintained. This array stores all ancestors of $u$ at distances equal to powers of 2, using the pointer jumping process. In the subsequent $\log n$ iterations, all vertices that lie on the path from $u_i$ to $r_i$ can be identified in parallel using special ancestors. For further details, we refer to \cite{pr_rst}.

\begin{lemma} \label{lemma: PR_RST_lemma}
\textnormal{(Lemma~1 and 2 from \cite{pr_rst})} Given a spanning forest $\hat{F}$ and a set $R$ of pairs, where each pair represents the end vertices of a disjoint path in $\hat{F}$, the paths represented by $R$ can be reversed using the broadcasting technique in $O(\log n)$ depth and $O(n \log n)$ work.
\end{lemma}

\textcolor{black}{The primary drawback of this technique is its $O(n \log n)$ space complexity. Additionally, it faces significant limitations when implemented on a \textsc{gpu}. The multi-step pointer jumping required to construct the special ancestor array involves multiple iterations and external synchronization, resulting in considerable runtime overhead. To overcome these inefficiencies, we propose a linear-space method in the next section that is also tailored for the \textsc{gpu} architecture.}

\subsection{Eulerian Tour Based Path Reversal}
Given a rooted forest $\hat{F}$ and a set $R$ of pairs representing the end vertices of paths in $\hat{F}$ whose orientation needs to be reversed, we design Algorithm~\ref{alg:reverse_path} to achieve this using \emph{Euler Tour}. 

In $\hat{F}$, let $firstEdge[v]$ and $lastEdge[v]$ denote the first and last edges incident on each vertex $v$. Each tree in $\hat{F}$ has its own Eulerian tour, where each tree edge appears twice.

As shown in Line~\ref{line:EulerianTour} of Algorithm~\ref{alg:reverse_path}, we first compute the Euler tours for all trees in $\hat{F}$ and assign ranks to the edges using the parallel algorithm by Polák et al. \cite{emc}. Based on these ranks, we assign \textit{start} and \textit{finish} times to each vertex, categorizing them into three groups as outlined in Lines~\ref{line:start_assign}--\ref{line:end_assign}. The assigned \textit{start} and \textit{finish} times for all vertices are shown in Figure~\ref{start_end_euler}.

For each path $(x_i, r_i)$, we update $support[r_i]$ with $x_i$ (Line~\ref{line:update_suppport}), enabling constant-time retrieval in subsequent steps. 
\textcolor{black}{In the last for loop of Algorithm~\ref{alg:reverse_path}, we go through all edges and identify the affected edges in Line~\ref{line:interval_condition}, and reverse their orientations.}

Each vertex $u$ knows its representative $\text{rep}[u]$ (the tree to which it belongs). Using this information, $u$ can easily determine the end vertex of the path (remember $support$ was updated earlier (Line~\ref{line:update_suppport})). Thus, $v$ can efficiently verify whether it lies on the path.

\begin{figure}[htbp]
\centering
\includegraphics[width=0.8\linewidth]{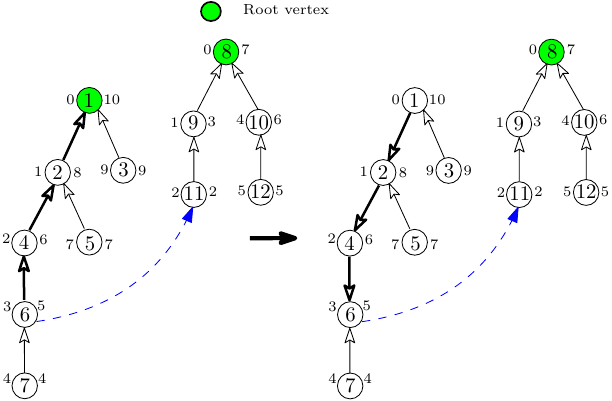}
\caption{Orientation of affected edges, show in solid thick, is reversed. Numbers shown on left and right sides of each vertex correspond to start and finish times, obtained using Euler-Tour.}
\label{start_end_euler}
\end{figure}

\begin{algorithm}[htbp]
    \DontPrintSemicolon
    \caption{\textsc{ReversePathsEulerTour}}
    \label{alg:reverse_path}
    \SetKwInput{KwInput}{Input}
    \SetKwInput{KwOutput}{Output}
    \KwInput{A rooted spanning forest $\hat{F}$, set of vertex pairs $R$, $rep[\cdot]$ and $parent[\cdot]$ arrays}
    \KwOutput{Updated $parent[\cdot]$ in the Forest $\hat{F}$}
    \SetKwBlock{DoParallel}{do in parallel}{end}

    $rank[\cdot] \gets$ \textsc{EulerianTour}($parent$)\;   \label{line:EulerianTour}

    \For{each vertex $v$ \textbf{in parallel}} { \label{line:start_assign}
        \If{$v$ is a \textbf{Leaf Vertex}} {
            $start[v] = finish[v] = rank[\langle v, parent[v] \rangle]$\;
        }
        \ElseIf{$v$ is a \textbf{Root Vertex}} {
            $start[v] = rank[firstEdge[v]]$\;
            $finish[v] = rank[lastEdge[v]] + 1$\;
        }
        \ElseIf{$v$ is an \textbf{Intermediate Vertex}} {
            $start[v] = rank[firstEdge[v]]$\;
            $finish[v] = rank[\langle v, parent[v] \rangle]$\;
        } \label{line:end_assign}
    } 

    \For{each vertex pairs $(x_i, r_i) \in R$ \textbf{in parallel}} { \label{line:firstLine}
        $support[r_i] = x_i$\; \label{line:update_suppport}
    }

\For{each edge $(u, parent[u])$ in $\hat{F}$ \textbf{in parallel}} { \label{line:start_reversal}
    $r = rep[u], x = support[r]$\; \label{line:support}
    \If{($start[r] < start[u] \leq start[x]$) and ($finish[r] > finish[u] \geq finish[x]$)} { \label{line:interval_condition}
        $v = parent[u]$ \;
        $parent[v] = u$ \; \label{line:lastLine}
    }
}
\end{algorithm}

\begin{theorem}
\label{theorem: ET}
Given a spanning forest $\hat{F}$ and a set $R$ of pairs, where each pair represents the end vertices of a disjoint path in $\hat{F}$, Algorithm~\ref{alg:reverse_path} reverses the paths represented by $R$ in $O(\log n)$ depth and $O(n)$ work.
\end{theorem}



\section{Correctness and Complexity Analysis}
In this section, we analyze the depth and work complexity of the four proposed algorithms. We then separately analyze the depth complexity for the insertion and deletion operations.

\begin{theorem}
    The \textsc{sg-bc} and \textsc{sg-et} algorithms \textbf{insert} a batch $B$ of edges in $O(\log n + \log^2 k)$ depth, while the \textsc{hs-bc} and \textsc{hs-et} algorithms require a depth of $O(\log^2 n)$. The total work performed by all algorithms is $O(n \log n)$.
\end{theorem}

\begin{theorem}
    The \textsc{sg-bc} and \textsc{sg-et} algorithms \textbf{delete} a batch $B$ of edge in  $O(\log n + \log^2 k)$ depth, while the \textsc{hs-bc} and \textsc{hs-et} algorithms require a depth of $O(\log^2 n)$. Both algorithms perform total work of $O(E + n \log n)$, where $n$ denotes the number of vertices and $E$ denotes the number of edges in $G$.
\end{theorem}


The space complexity of the Algorithm ~\ref{alg:main_algo} is bounded by $O(V + E)$, where $V$ represents the number of vertices and $E$ the number of edges in the graph. This complexity arises from storing the graph (using an edge list) and auxiliary structures such as the parent array, representative (rep) array, and similar data structures. Temporary storage for operations such as compaction, sorting or hash table also scales with $V$ or $E$. Additionally, $O(B)$ space is required for handling a batch $B$ of edges.

Further, broadcasting approach requires $O(V\log V)$ space as each vertex stores upto $\log V$ ancestors \cite{pr_rst}. Consequently, the \textsc{sg-bc} and \textsc{hs-bc} algorithms consume $O(V \log V + E)$ space, whereas \textsc{sg-et} and \textsc{hs-et} algorithms consume $O(V + E)$ space.

We now proceed to prove the correctness of Algorithm~\ref{alg:main_algo} and Algorithm~\ref{alg:reverse_path}.

\begin{lemma}
Let $R$ be a set containing $k$ pairs of vertices $(x_1, y_1), \allowbreak (x_2, y_2), \allowbreak \ldots, \allowbreak (x_k, y_k)$ in a forest $F = \{T_1, \allowbreak T_2, \allowbreak \ldots, \allowbreak T_k\}$, where $y_i$ is the root of a tree $T_i \in F$ and $x_i$ is a vertex in $T_i$. Using the $start$ and $end$ times obtained from the Eulerian tour, Algorithm~\ref{alg:reverse_path} correctly reverses all paths from each $x_i$ to its corresponding root $y_i$.
\end{lemma}

\begin{proof}
    Consider a tree $T_i$ and vertex pair $(x_i,y_i)$. Let $U$ to be the set of vertices $u_i$ such that $first[x_i] < first[u_i] \leq first[y_i]$ and $V$ to be the set of vertices $v_i$ such that $last[x_i] > last[v_i] \geq last[y_i]$. Then $Path(x_i,y_i)$ = $\{ w_i \mid w_i \in U \cap V \}$ consists of all vertices in the simple path from $x_i$ to $y_i$. $\forall$ $i,j \in [1,k]$ such that $i \neq j$, since any two trees $T_i$, $T_j$ are disjoint, it follows that $Path(x_i,y_i) \cap Path(x_j,y_j) = \emptyset$. Thus, all paths can be reversed without conflicts.

\end{proof}
\section{Experiments and Results}
In this section, we introduce the datasets and the configuration of the experimental platform. We then analyse the performance and scalability of the proposed algorithms compared to the static \textsc{bfs} and \textsc{pr-rst} algorithms. Finally, we highlight key insights and the influencing factors of the \textsc{rsf} algorithm.

\subsection{Experimental Setup and Dataset}
We conduct our experiments on an NVIDIA A100 GPU with 80\,GB of on-board memory, based on the Ampere architecture. The GPU features 6912 CUDA cores, 80 Streaming Multiprocessors (SMs), 2.04\,TB/s memory bandwidth, 6\,MB of L2 cache, and 128\,KB of L1 cache per SM. The GPU is connected to an AMD EPYC 7742 CPU with 64 cores, 4\,MB of L1 cache, 32\,MB of L2 cache, and 256\,MB of L3 cache. The number of threads per block was set to 1024 for all experiments. To facilitate reproducibility and future comparisons, we have made the complete source code used for the experiments, along with all implementation details publicly available.\footnote{\url{https://github.com/Abhijeetkumar96/batch-dynamic-spanning-tree.git}}

To test the capability of our algorithms under dynamic graph updates, we used 13 datasets that include both real-world and synthetic graphs, as listed in Table~\ref{tab:graph-stats}. Sourced from various publicly available repositories, these datasets represent road networks, web graphs, co-purchasing networks, and social networks \cite{boldi2011layered, leskovec2016snap}, and range from 1 million to 1 billion edges. They are commonly used in evaluations of \textsc{gpu}-based graph algorithms \cite{emc, chirayu_bcc, hooking_shortcutting}. Prior to our experiments, we preprocessed each graph by treating directed edges as undirected and removing multiple edges and self-loops, similar to other works.

Finally, to test the scalability of our algorithms, we check their ability to handle different batch update sizes. We use a wide range of batch update sizes, ranging from 100 to 100,000 edges. The batch edges are selected uniformly at random from the graph. 

A delete batch created uniformly at random, likely to have fewer tree edges compared to non-tree edges, and thus we choose (approximately 30 - 40\%) of the edges from tree edges and the rest from the non-tree edges for a delete batch.

For each graph and each batch size, we generate five independent instances. Each instance is executed five times independently, totaling 25 individual executions. The speedup is computed by averaging the ratios of baseline execution time to our method’s execution time over these 25 runs. While multiple batch sizes were evaluated to analyze performance trends, we report results primarily for batch size 10,000 for brevity.

\begin{table}[htbp]
\centering
\caption{Statistics of graphs used in the experiments}
\label{tab:graph-stats}
\begin{tabular}{cccc}
\hline
\textbf{Dataset} & \textbf{V} & \textbf{E} & \textbf{Abv.} \\
\hline

\multicolumn{4}{c}{\textbf{Social Networks}} \\
\hline
higgs-twitter       & 456K      & 14.8M     & TW \\
hollywood-2009      & 1.1M      & 56.3M     & HW \\
com-Orkut           & 3.07M     & 117M      & CO \\
soc-LiveJournal1    & 4.8M      & 68M       & LJ \\
\hline

\multicolumn{4}{c}{\textbf{Web Graphs}} \\
\hline
uk-2002             & 18.5M     & 298M      & UK-02 \\
arabic-2005         & 22.7M     & 639M      & AR    \\
uk-2005             & 39.4M     & 936M      & UK-05 \\
\hline

\multicolumn{4}{c}{\textbf{Road Networks}} \\
\hline
Road-USA            & 23.9M     & 58M       & RU \\
europe\_osm         & 50.9M     & 54.1M     & EO \\
\hline

\multicolumn{4}{c}{\textbf{Kronecker (Synthetic) Graphs}} \\ 
\hline
kron\_g500-logn18   & 262.1K    & 10.6M     & KR-18 \\
kron\_g500-logn19   & 524K      & 22M       & KR-19 \\
kron\_g500-logn20   & 1.0M      & 45M       & KR-20 \\
kron\_g500-logn21   & 2.1M      & 91M       & KR-21 \\
\hline

\end{tabular}
\end{table}

\begin{figure}[ht]
\centering
\includegraphics[width=1\linewidth]{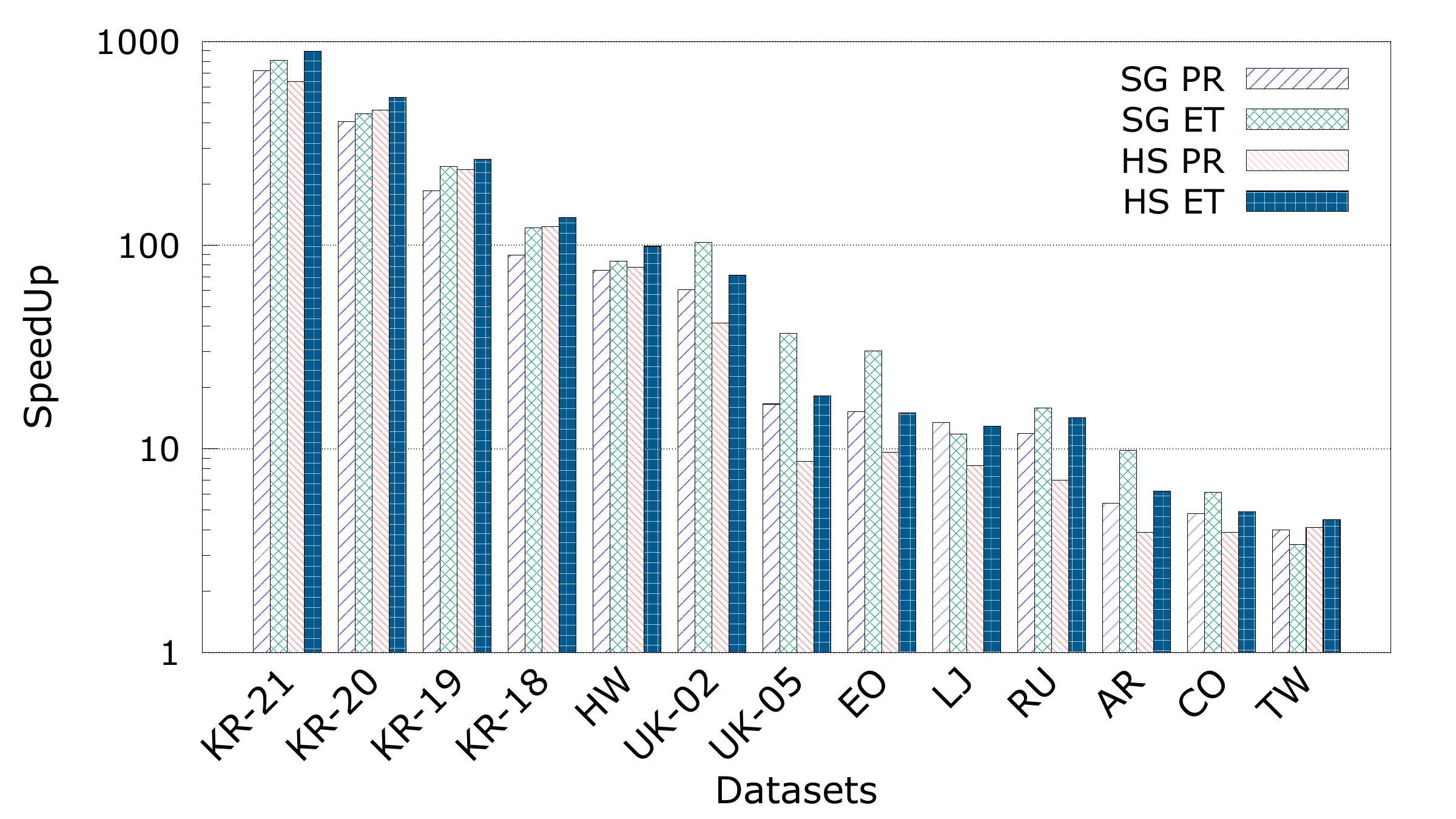}
\caption{Speedup: Our algorithms w.r.t static \textsc{gpu bfs} for insert}
\label{insertion_speedup}
\end{figure}

\begin{figure}[ht]
\centering
\includegraphics[width=1\linewidth]{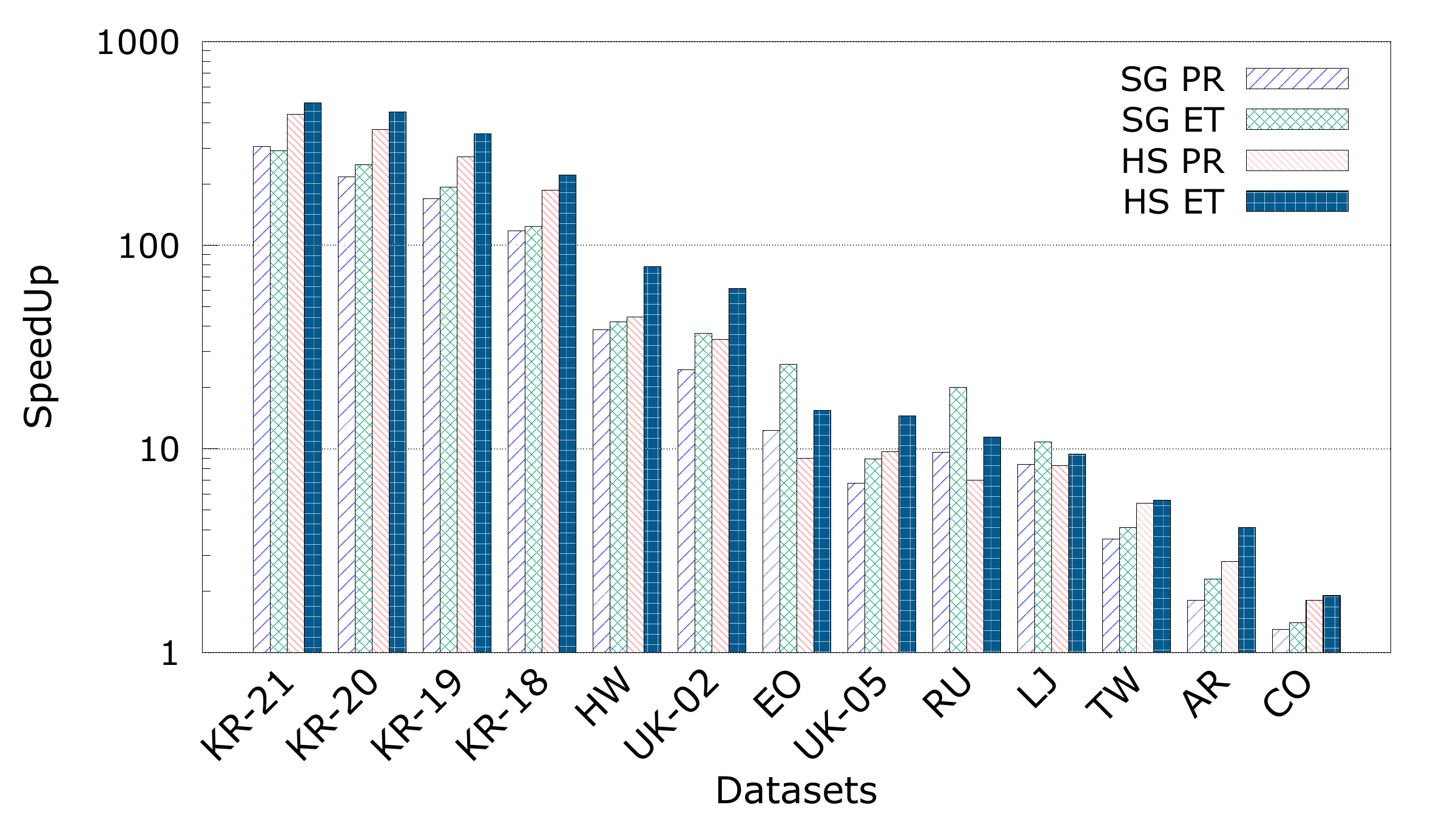}
\caption{Speedup: Our algorithms w.r.t static \textsc{gpu bfs} for delete}
\label{deletion_speedup}
\end{figure}

\begin{figure*}[h!]
\centering
\includegraphics[width=0.9\linewidth]{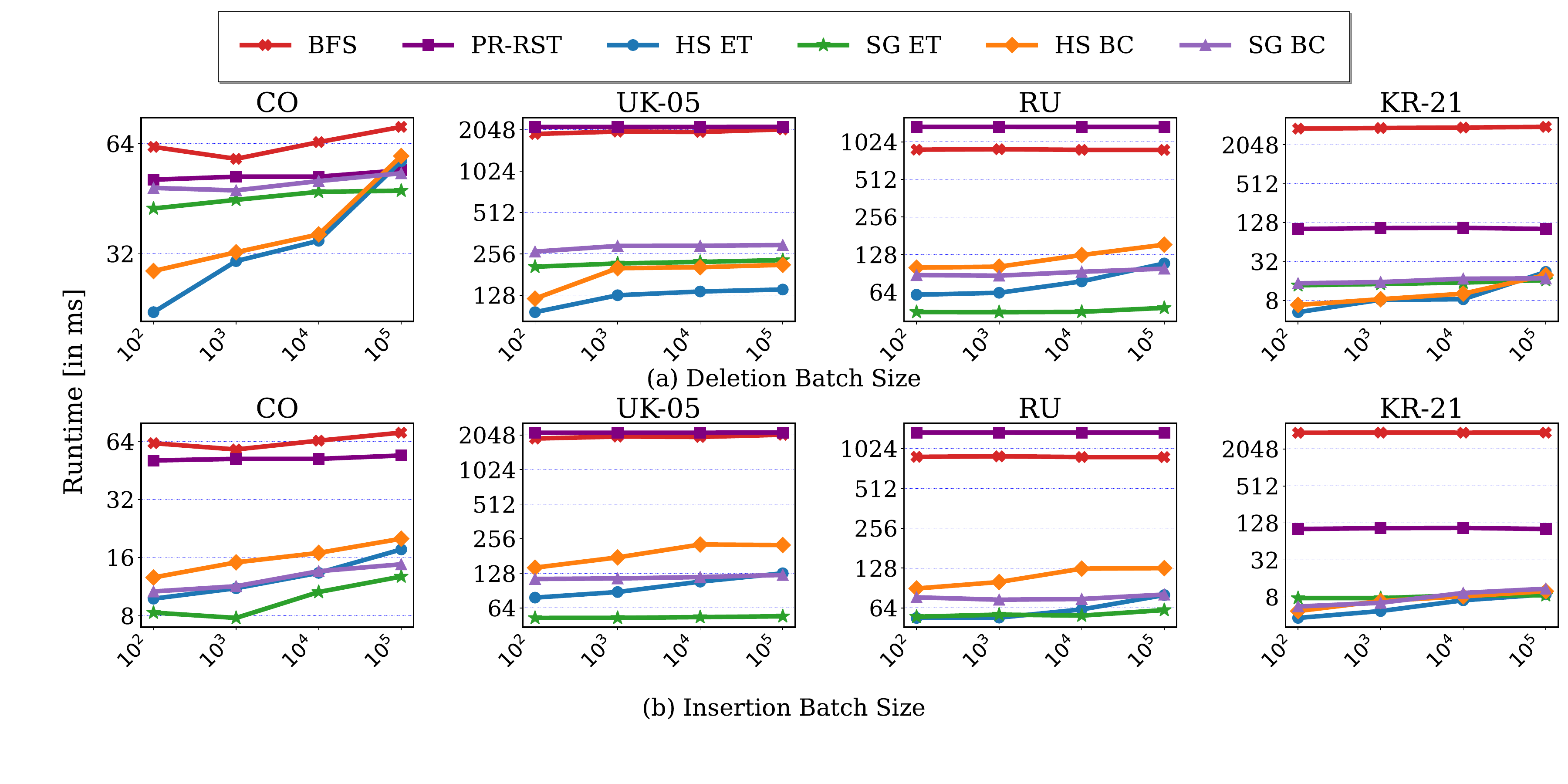}

\caption{Runtime comparison of baseline and proposed algorithms across various batches for (a) Deletion operation and (b) Insertion operation.}

\label{deletion_combined}
\end{figure*}

\subsection{Comparison with Static Parallel Algorithms}

In this section, we compare the performance of our proposed algorithms against two static approaches: the parallel \textsc{bfs} algorithm, which is a well-known and commonly used baseline for graph traversal, and the multi-core \textsc{pr-rst} algorithm proposed by Cong et al.~\cite{pr_rst}.

\textbf{Comparison with Parallel \textsc{BFS}.}
Parallel \textsc{bfs} algorithms typically assume that the graph is connected; however, when it is not, \textsc{bfs} must be initiated separately for each disconnected component, leading to multiple sequential executions. To eliminate such serial processing, we first add or delete the necessary batch of edges and then identify all disconnected components using the method described in~\cite{hooking_shortcutting}. We subsequently use stream compaction to select all unique representatives and apply parallel \textsc{bfs} starting from these representatives.

Our results show a maximum speedup of $530\times$ for edge deletion (with an average speedup of $160\times$) and $900\times$ for edge insertion (with an average speedup of $200\times$), as illustrated in Figures~\ref{insertion_speedup} and~\ref{deletion_speedup}.

The observed improvements in running times can be attributed to the dynamic setting.
The size of the largest connected component and the depth of the tree continually change as new batches of edges are inserted or deleted. This variability affects the performance of \textsc{bfs} due to its dependence on the graph's diameter~\cite{emc}. However, all of our proposed algorithms, being independent of the tree's depth, demonstrated significant speedup compared to the baseline, making them more suitable for real-world dynamic graphs.

\textbf{Comparison with  \textsc{pr-rst}.}
\textsc{pr-rst} is a parallel algorithm proposed by by Cong et al. \cite{pr_rst} for multi-core \textsc{cpu}s. While \textsc{bfs} has been extensively studied and widely implemented, the multi-core \textsc{pr-rst} algorithm has not yet seen widespread adoption particularly on \textsc{gpu}s, to the best of our knowledge. To enable a fair comparison, we ported the original multi-core implementation to the \textsc{gpu} ourselves.

As shown in Figure~\ref{deletion_combined}, our best-performing algorithm achieves, on average, a $13\times$ speedup over \textsc{pr-rst} for deletions (maximum $30\times$) and an $18\times$ speedup for insertions (maximum $41\times$). This performance gap is largely due to the fact that \textsc{pr-rst} is a static algorithm that recomputes the entire rooted spanning forest after every update, without leveraging prior computations. Additionally, the construction of the ancestor array (Section~\ref{BC_path}) results in numerous uncoalesced memory accesses, which further degrade \textsc{gpu} performance.

\subsection{Scalability Analysis: Impact of Increasing Batch Sizes on Performance}
After evaluating the speedups of our proposed algorithms, we now focus on studying their scalability. Specifically, our goal is to understand how the performance of the algorithms changes with increasing batch sizes.

We conducted experiments using batch sizes of 100, 1000, 10,000, and 100,000 edges on the largest graph of each category of datasets in Table~\ref{tab:graph-stats} and running times are reported in Figure ~\ref{deletion_combined}.
The results demonstrate that our algorithms maintain strong scalability, with minimal variation in running times across different batch sizes, achieving an average throughput of 2M insertions and 1.4M deletions per second across all batch sizes respectively. This suggests that our algorithms can efficiently handle larger batches without significant performance degradation.

\begin{figure*}[!tbh]
  \centering
  \begin{subfigure}[b]{0.325\textwidth}
    \includegraphics[width=\textwidth]{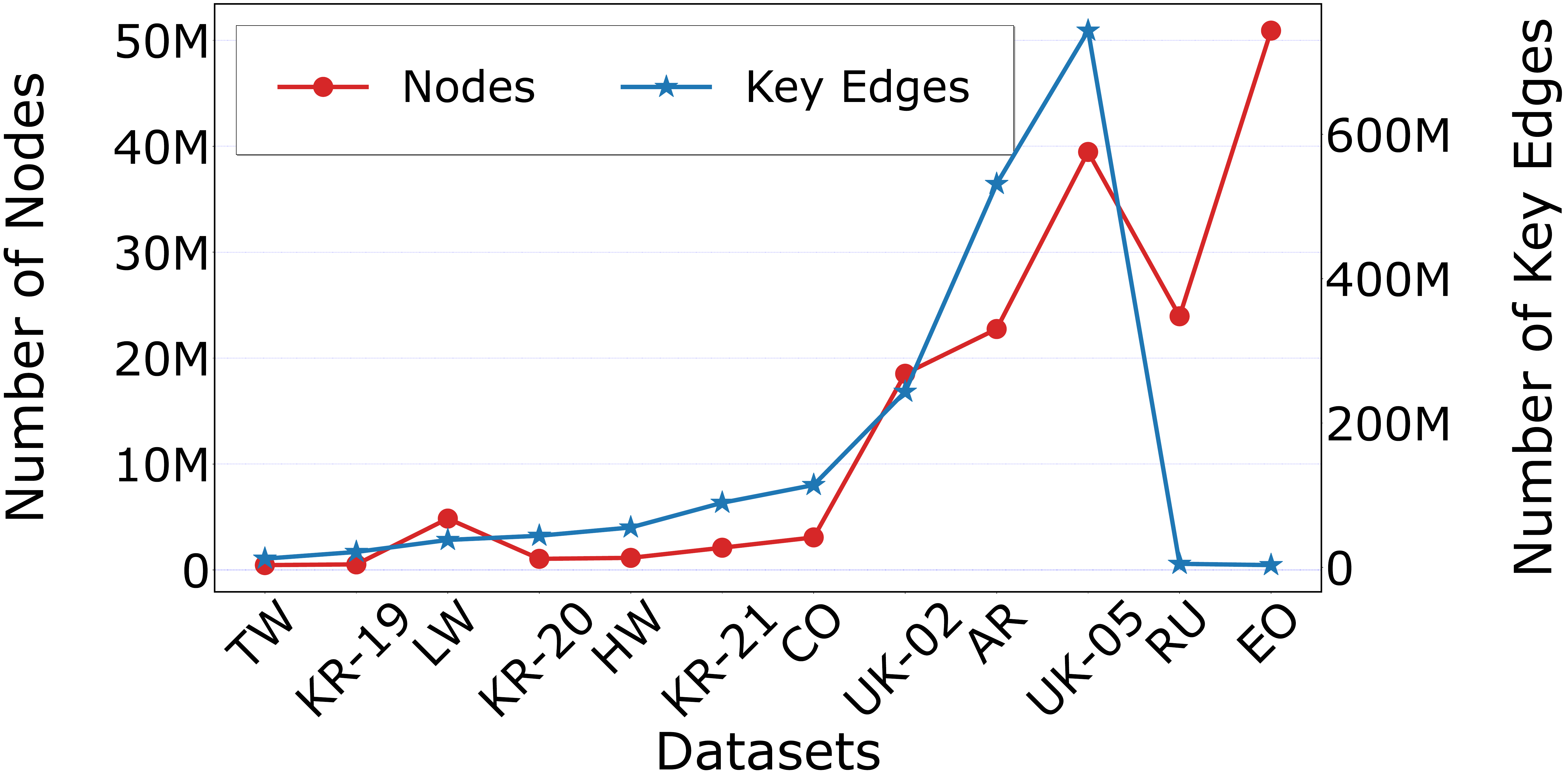}
    \caption{Vertices vs.\ key edges for deletion (batch size: 10,000).}
    \label{fig:behaviour-plot-1}
  \end{subfigure}\hfill
  \begin{subfigure}[b]{0.315\textwidth}
    \includegraphics[width=\textwidth]{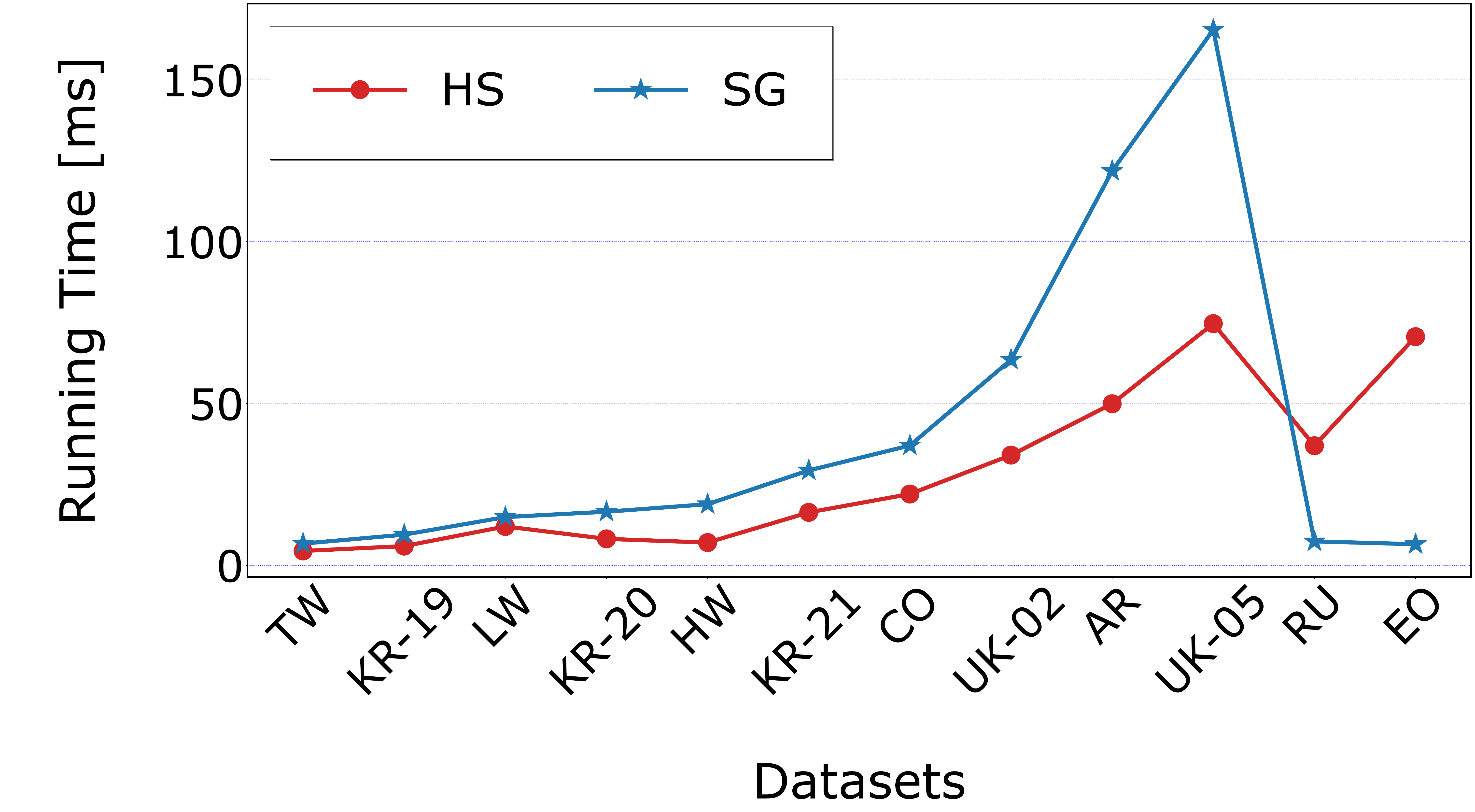}
    \caption{Running times for deletion operations.}
    \label{fig:behaviour-plot-2}
  \end{subfigure}\hfill
  \begin{subfigure}[b]{0.315\textwidth}
    \includegraphics[width=\textwidth]{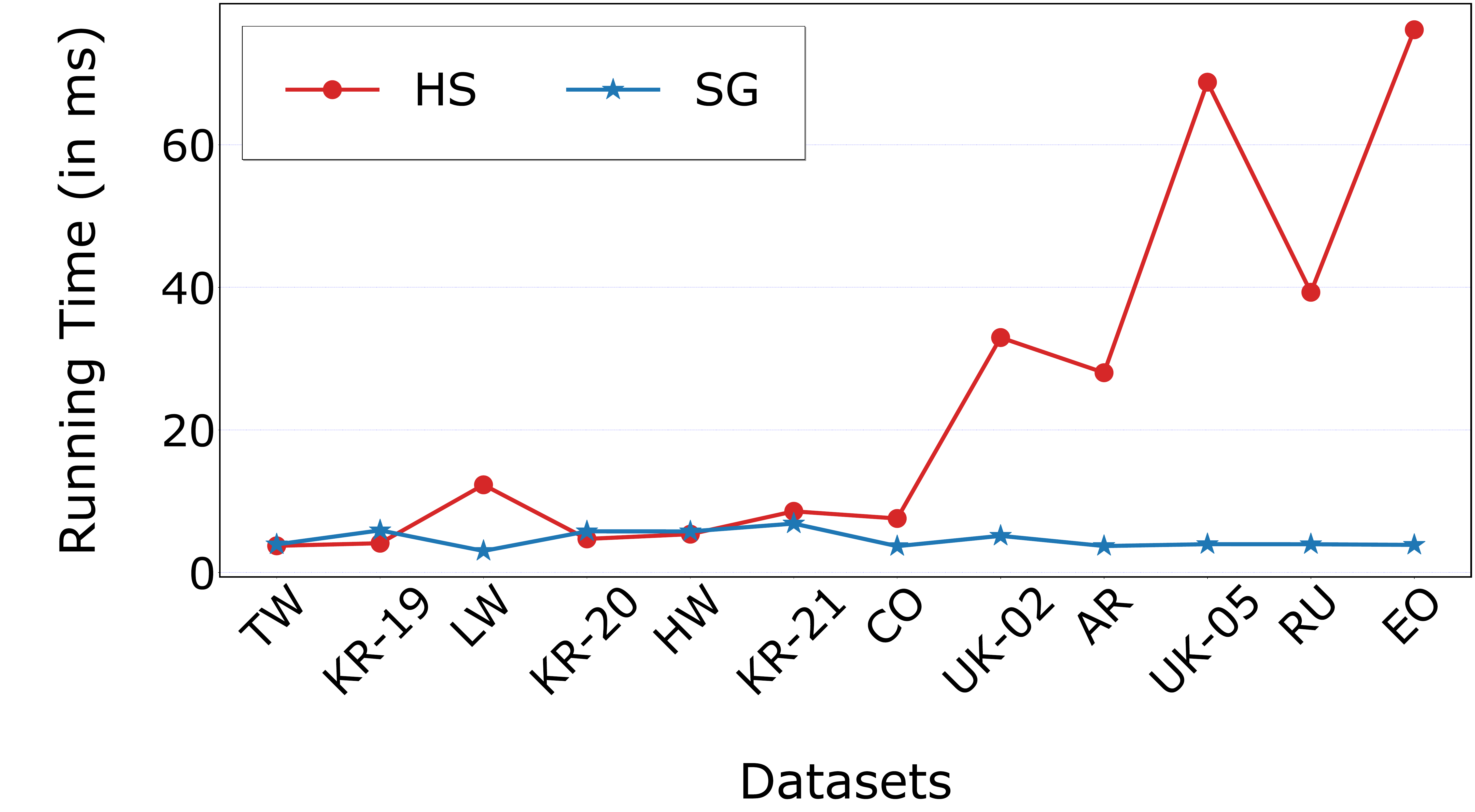}
    \caption{Running times for insertion operations.}
    \label{fig:behaviour-plot-3}
  \end{subfigure}

  \caption{HS performance depends on the number of nodes ($n$), whereas SG performance depends on the number of key edges ($m$). For insertion,  the number of key edges equals the batch size, which is fixed across datasets.}
  \label{fig:behaviour-analysis-combined}
\end{figure*}

\subsection{Influencing Factors in proposed algorithms}
In this section, we analyze the factors influencing the performance of our algorithms. While all algorithms demonstrated significant speedups and scalability, running times varied across different batches and operations. We now identify the key factors that contribute to these variations. 

Although deletions theoretically introduce an $O(\log E)$ overhead, our experiments reveal that this accounts for only 10-15\% of the total runtime. This is largely because the compaction algorithm is highly optimized for GPUs due to its effective parallelization techniques. However, most of the time (approximately 50\%) was spent on identifying oriented replacement edges, while path reversal accounted for around 30\%, except for insertions in SG-based algorithms. Interestingly, despite both broadcasting and the Eulerian Tour having the same $O(\log n)$ complexity for $Path\ Reversal$, the latter demonstrated superior practical performance. We now explore the rationale behind these observed behaviors.

\noindent \textbf{Hooking-Shortcutting vs SuperGraph.}
For \textit{deletions} (Figures~\ref{fig:behaviour-plot-1},~\ref{fig:behaviour-plot-2}), \textsc{hs} runtime scales with the number of vertices $n$, while \textsc{sg} scales with the number of key edges $m$. Consequently, \textsc{hs} outperforms \textsc{sg} when $n < m$, and vice versa. This distinction becomes critical for dense graphs, where $m$ can reach $O(n^2)$, making \textsc{sg} increasingly inefficient while \textsc{hs} maintains stable performance.

For \textit{insertions} (Figure~\ref{fig:behaviour-plot-3}), \textsc{sg} runtime remains nearly constant since the number of key edges equals the fixed batch size, whereas \textsc{hs} runtime continues to vary with $n$. As a result, \textsc{sg} consistently outperforms \textsc{hs} for insertions, even on dense graphs, since only the inserted batch is processed.

In summary, \textsc{hs} is preferable when $n \ll m$, while \textsc{sg} is preferable when $m \ll n$.

Regarding path reversal, we notice that the Eulerian Tour technique outperforms Broadcasting by a factor of $2-3\times$. This performance gap is primarily due to uncoalesced memory accesses in Broadcasting and number the space complexity involved. In particular, Eulerian Tour technique requires $O(n)$ space, where as Broadcasting requires $O(n\log n)$.
\section{Conclusion}  
We presented four algorithms for the fully dynamic update of rooted spanning trees in a many-core (GPU) environment. Our algorithms focus on repairing the spanning tree rather than reconstructing it from scratch, using key techniques such as hooking and shortcutting, supergraph-based methods for identifying oriented replacement edges, and Eulerian tour technique for path reversal. Our results demonstrate that incremental tree maintenance is a highly effective strategy for processing massive, evolving graphs on modern accelerators.

\section{Future Work}
Several promising directions for future research remain. First, we plan to extend our algorithm to support \textit{fully dynamic connectivity} beyond spanning trees. Second, we aim to optimize the dynamic updates of auxiliary data structures, such as Euler Tour Trees, to enable more efficient batch-parallel queries. Finally, we intend to explore multi-GPU implementations to further scale performance for massive, evolving real-world graphs.

\section*{Acknowledgment}
The authors would like to thank the anonymous reviewers for their insightful feedback. We also thank Dr. Dip Sankar Banerjee (IITJ) for providing access to the DGX system support used in this work.

\bibliographystyle{IEEEtran}
\bibliography{text/References}

\end{document}